\documentclass[twocolumn,english]{IEEEtran}
\usepackage[T1]{fontenc}
\usepackage[latin9]{inputenc}
\usepackage{float}
\usepackage{bm}
\usepackage{amsmath}
\usepackage{graphicx}
\usepackage{amssymb}

\makeatletter

\floatstyle{ruled}
\newfloat{algorithm}{tbp}{loa}
\floatname{algorithm}{Algorithm}

\newtheorem{prop}{Proposition}

\usepackage{cite}
\usepackage{bm}
\usepackage{bbm}
\usepackage{mathrsfs}

\@ifundefined{showcaptionsetup}{}{%
 \PassOptionsToPackage{caption=false}{subfig}}
\usepackage{subfig}
\makeatother

\usepackage{babel}

\begin{document}

\title{Subspace Evolution and Transfer (SET) for Low-Rank Matrix Completion}

\author{Wei Dai$^{1}$, Olgica Milenkovic$^{1}$ and Ely Kerman$^{2}$\\
$^{1}$Department of Electrical and Computer Engineering, $^{2}$Department
of Mathematics\\
University of Illinois at Urbana-Champaign\\
 Email: \{weidai07,milenkov, ekerman\}@illinois.edu}
\maketitle
\begin{abstract}
We describe a new algorithm, termed subspace evolution and transfer
(SET), for solving low-rank matrix completion problems. The algorithm
takes as its input a subset of entries of a low-rank matrix, and outputs
\emph{one} low-rank matrix consistent with the given observations.
The completion task is accomplished by searching for a column space
on the Grassmann manifold that matches the incomplete observations.
The SET algorithm consists of two parts -- subspace evolution and
subspace transfer. In the evolution part, we use a gradient descent
method on the Grassmann manifold to refine our estimate of the column
space. Since the gradient descent algorithm is not guaranteed to converge,
due to the existence of barriers along the search path, we design
a new mechanism for detecting barriers and transferring the estimated
column space across the barriers. This mechanism constitutes the core
of the transfer step of the algorithm. The SET algorithm exhibits
excellent empirical performance for both high and low sampling rate
regimes. \end{abstract}
\thispagestyle{empty}

\renewcommand{\thefootnote}{\fnsymbol{footnote}} \footnotetext[0]{The authors would like to thank Dayu Huang for his help in designing the employed line-search procedure, and to acknowledge useful discussions with Yoram Bresler, Justin Haldar, Angelia Nedich, and Zoi Rapti. Furthermore, the authors would also like to thank the authors of  for providing online software packages for their matrix completion algorithms. Parts of the results in the paper were presented at ICASSP 2010, Dallas, Texas.} \renewcommand{\thefootnote}{\arabic{footnote}} \setcounter{footnote}{0}

\section{Introduction}

Suppose that we observe a subset of entries of a matrix. The matrix
completion problem asks when and how the matrix can be recovered based
on the observed entries. In general, this reconstruction task is ill-posed
and computationally intractable. However, if the data matrix is known
to have low-rank, exact recovery can be accomplished in an efficient
manner with high probability, provided that sufficiently many entries
are revealed. Low-rank matrix completion problems have received considerable
interests due to their wide applications, ranging from collaborative
filtering (the NETFLIX challenge) to sensor network tomography. For
an overview of these applications, the reader is referred to \cite{Candes_Recht_2008_matrix_completion}.

An efficient way to solve the completion problem is via convex relaxation.
Instead of looking at rank-restricted matrices, one can search for
a matrix with minimum nuclear norm, subject to data consistency constraints.
Although in general nuclear norm minimization is not equivalent to
rank minimization, the former approach recovers the same solution
as the latter if the data matrix satisfies certain incoherence conditions
\cite{candes_tao_power_2009}. More importantly, nuclear norm minimization
can be accomplished in polynomial time by using semi-definite programming,
singular value thresholding (SVT) \cite{cai_singular_2008}, or methods
adapted from robust principal component analysis \cite{Candes2009_robust_principal_component}.

Several low-complexity alternatives to nuclear norm minimization have
been proposed so far. Realizing the intimate relationship between
compressive sensing and low-rank matrix completion, a few approaches
for low-rank completion can be viewed as generalization of those for
compressive sensing reconstruction. In particular, the ADMiRA algorithm
\cite{lee_Bresler_admira:_2009} is a counterpart of the subspace
pursuit (SP) \cite{Dai_2009_Subspace_Pursuit} and CoSaMP \cite{Tropp_2009_CoSaMP}
algorithms, while the singular value projection (SVP) method \cite{Meka2009_SVP}
extends the iterative hard thresholding (IHT) \cite{Blumensath_Davies_2009_IHT}
approach. There are other approaches that rely more on the specific
structures of the low-rank matrices. The power factorization algorithm
described in \cite{Haldar_Hernando_powerfactorization_2009} takes
an alternating optimization approach. In the OptSpace algorithm described
in \cite{montanari_keshavan_matrix_2009}, a simultaneous optimization
on both column and row spaces is employed. 

We address a more general class of problems in low-rank matrix completion
-- \emph{consistent} completion. Consistent completion extends the
previous completion framework in that it does not require the existence
of a unique solution to the problem. This extension seems questionable
at first glance -- in highly \emph{undersampled} observation regimes,
there may exist many low-rank matrices that match the observations
-- which makes the final result have less practical value. Nevertheless,
the consistent completion paradigm allows for identifying convergence
problems with standard completion techniques, and it does not require
any additional structure on the matrix, such as incoherence. Furthermore,
as will be shown in the subsequent exposition, when confronted with
very sparsely sampled matrices \emph{all methods known so far} fail
to produce any solution to the problem, despite the fact that many
exist. Finally, even in the sampling regime for which SVT, OptSpace
and other techniques have provable, unique reconstruction performance
guarantees, the consistent completion technique described in this
contribution exhibits significantly better results.

To solve the consistent matrix completion problem, we propose a novel
subspace evolution and transfer (SET) method. We show that the matrix
completion problem can be solved by searching for a column space (or,
alternatively, for a row space) that matches the observations. As
a result, optimization on the Grassmann manifold, i.e., subspace evolution,
plays a central role in the algorithm. However, there may exist {}``barriers''
along the search path that prevent subspace evolution from converging
to a global optimum. To address this problem, in the subspace transfer
part, we design mechanisms to detect and cross barriers. The SET algorithm
improves the recovery performance not only in high sampling rate regime
but also in low sampling rate regime where there may exist many low-rank
solutions. Empirical simulations demonstrate the excellent performance
of the proposed algorithm.

The SET algorithm employs a similar approach as that of the OptSpace
algorithm \cite{montanari_keshavan_matrix_2009} in terms of using
optimization over Grassmann manifolds. Still, the SET approach substantially
differs from the method supporting OptSpace \cite{montanari_keshavan_matrix_2009}.
Searching over only one space (column or row space) represents one
of the most significant differences: in OptSpace, one searches \emph{both}
column and row spaces simultaneously, which introduces numerical and
analytical difficulties. Moreover, when optimizing over the column
space, one has to take care of {}``barriers'' that prevent the search
procedure from converging to a global optimum, an issue that was not
addressed before since it was obscured by simultaneous column and
row space searches.

The paper is organized as follows. In Section~\ref{sec:consistent}
we introduce the consistent low-rank completion problem, and describe
the terminology used throughout the paper. In Section~\ref{sec:SET}
we outline the steps of the SET algorithm. Simulation results are
presented in Section~\ref{sec:Simulations}. All proofs are listed
in the Appendix sections.

\vspace{-0.06in}

\section{\label{sec:consistent}Consistent Matrix Completion}

Let $\bm{X}\in\mathbb{R}^{m\times n}$ be an unknown matrix with rank
$r\ll\min\left(m,n\right)$, and let $\Omega\subset\left[m\right]\times\left[n\right]$
be the set of indices of the observed entries, where $\left[K\right]=\left\{ 1,2,\cdots,K\right\} $.
Define the projection operator $\mathcal{P}_{\Omega}$ by \begin{align*}
\mathcal{P}_{\Omega}:\;\mathbb{R}^{m\times n} & \rightarrow\mathbb{R}^{m\times n}\\
\mathcal{P}_{\Omega}(\bm{X}) & \mapsto\bm{X}_{\Omega},\;\mbox{where }\left(\bm{X}_{\Omega}\right)_{i,j}=\begin{cases}
\bm{X}_{i,j} & \mbox{if }\left(i,j\right)\in\Omega\\
0 & \mbox{if }\left(i,j\right)\notin\Omega\end{cases}.\end{align*}
 The \emph{consistent matrix completion} problem is to find \emph{one}
rank-$r$ matrix $\bm{X}^{\prime}$ that is consistent with the observations
$\bm{X}_{\Omega}$, i.e., \begin{align}
\left(P0\right):\; & \mbox{find a }\bm{X}^{\prime}\mbox{ such that }\nonumber \\
 & \mbox{rank}\left(\bm{X}^{\prime}\right)\le r\mbox{ and }\mathcal{P}_{\Omega}\left(\bm{X}^{\prime}\right)=\mathcal{P}_{\Omega}\left(\bm{X}\right)=\bm{X}_{\Omega}.\label{eq:P0}\end{align}
 This problem is well defined as all our instances of $\bm{X}_{\Omega}$
are generated from matrices $\bm{X}$ with rank $r$ and therefore
there must exist at least one solution. Here, like in other approaches
\cite{lee_Bresler_admira:_2009,Haldar_Hernando_powerfactorization_2009,montanari_keshavan_matrix_2009},
we assume that the rank $r$ is given. In practice, one may try to
sequentially guess a rank bound until a satisfactory solution has
been found.

We also introduce the (standard) projection operator $\mathcal{P}$,
\begin{align*}
\mathcal{P}:\;\mathbb{R}^{m}\times\mathbb{R}^{m\times k} & \rightarrow\mathbb{R}^{m}\\
\mathcal{P}\left(\bm{x},\bm{U}\right) & \mapsto\bm{y}=\bm{U}\bm{U}^{\dagger}\bm{x},\end{align*}
 where $1\le k\le m$, and where the superscript $\dagger$ denotes
the pseudoinverse of a matrix. That is, $\mathcal{P}\left(\bm{x},\bm{U}\right)$
gives the projection of the vector $\bm{x}$ on the hyperplane spanned
by the matrix $\bm{U}$, i.e., $\mbox{span}\left(\bm{U}\right)$.
It should be observed that $\bm{U}^{\dagger}\bm{x}$ is the global
minimizer of the quadratic optimization problem $\min_{\bm{w}\in\mathbb{R}^{k}}\;\left\Vert \bm{x}-\bm{U}\bm{w}\right\Vert _{F}^{2}.$

\subsection{Why optimizing over column spaces only?}

In this section, we show that the problem $\left(P0\right)$ is equivalent
to finding a column space consistent with the observations.

Let $\mathcal{U}_{m,r}$ be the set of $m\times r$ matrices with
$r$ orthonormal columns, i.e., $\mathcal{U}_{m,r}=\left\{ \bm{U}\in\mathbb{R}^{m\times r}:\;\bm{U}^{T}\bm{U}=\bm{I}_{r}\right\} .$
Define a function \begin{align}
f:\;\mathcal{U}_{m,r} & \rightarrow\mathbb{R}\nonumber \\
f(\bm{U}) & \mapsto\underset{\bm{W}\in\mathbb{R}^{n\times r}}{\min}\left\Vert \bm{X}_{\Omega}-\mathcal{P}_{\Omega}\left(\bm{U}\bm{W}^{T}\right)\right\Vert _{F}^{2},\label{eq:f_U}\end{align}
 where $\left\Vert \cdot\right\Vert _{F}$ denotes the Frobenius norm.
The function $f$ captures the consistency between the matrix $\bm{U}$
and the observations $\bm{X}_{\Omega}$ : if $f\left(\bm{U}\right)=0$,
then there exists a matrix $\bm{W}$ such that the rank-$r$ matrix
$\bm{U}\bm{W}^{T}$ satisfies $\mathcal{P}_{\Omega}\left(\bm{U}\bm{W}^{T}\right)=\bm{X}_{\Omega}$.
Hence, the consistent matrix completion problem is equivalent to \begin{align}
\left(P1\right):\; & \mbox{find }\bm{U}\in\mathcal{U}_{m,r}\;\mbox{such that }f\left(\bm{U}\right)=0.\label{eq:P1}\end{align}

An important property of the objective function $f$ is that $f$
is invariant under rotations. More precisely, $f\left(\bm{U}\right)=f\left(\bm{U}\bm{V}\right)$
for any $r$-by-$r$ orthogonal matrix $\bm{V}\in\mathcal{U}_{r,r}$.
This can be easily verified, as $\bm{U}\bm{W}^{T}=\left(\bm{U}\bm{V}\right)\left(\bm{W}\bm{V}\right)^{T}$.
Hence, the function $f$ depends only on the subspace spanned by the
columns of $\bm{U}$, i.e., the $\mbox{span}\left(\bm{U}\right)$.
Note that all columns of the matrix of the form $\bm{U}\bm{W}^{T}$
lie in the linear subspace $\mbox{span}\left(\bm{U}\right)$. The
consistent matrix completion problem essentially reduces to finding
a column space consistent with the observed entries. Note that instead
of identifying the column space in which the observations lie, one
can also use the row space instead. All results and the problem formulation
remain valid in this case as well. Which space to search over will
depend on the dimension of the matrix, and the particular sampling
pattern (which determines the density of rows and columns of the matrix).
In addition, one can run in parallel two search procedures - one on
the column space, the other on the row space. Here, we only focus
on the simplest scenario, and restrict our attention to column spaces.

\subsection{Grassmann manifolds and geodesics}

We find the following definitions useful for the exposition to follow.
The Grassmann manifold $\mathcal{G}_{m,r}$ is the set of all $r$-dimensional
linear subspaces (hyperplanes through the origin) in $\mathbb{R}^{n}$,
i.e., $\mathcal{G}_{m,r}=\left\{ \mbox{span}\left(\bm{U}\right):\;\bm{U}\in\mathcal{U}_{m,r}\right\} $.
Given a subspace $\mathscr{U}\in\mathcal{G}_{m,r}$, one can always
find a matrix $\bm{U}\in\mathcal{U}_{m,r}$, such that $\mathscr{U}=\mbox{span}\left(\bm{U}\right)$.
The matrix $\bm{U}$ is referred to as a generator matrix of $\mathscr{U}$
and the columns of $\bm{U}$ are often referred to as an orthonormal
basis of $\mathscr{U}$. Since $\mbox{span}\left(\bm{U}\right)=\mbox{span}\left(\bm{U}\bm{V}\right)$
for all $\bm{V}\in\mathcal{U}_{r,r}$, it is clear that the generator
matrix for a given subspace is not unique. Nevertheless, a given matrix
$\bm{U}\in\mathcal{U}_{m,r}$ uniquely defines a subspace. For this
reason, we henceforth use $\bm{U}$ to represent its induced subspace.

To search for a consistent column space, we use a gradient descent
method on the Grassmann manifold. For this purpose, we introduce the
notion of a geodesic curve in the Grassmann manifold. Roughly speaking,
a geodesic curve is an analogue of a straight line in an Euclidean
space: given two points on the manifold, the geodesic curve connecting
them is the path of the shortest length in the manifold. Let $\bm{U}\left(t\right)$
be a geodesic curve (parametrized by $t\in\mathbb{R}$) in the Grassmann
manifold. Denote the starting point of this geodesic curve by $\bm{U}\left(0\right)=\bm{U}\in\mathcal{U}_{m,r}$,
and the direction by $\dot{\bm{U}}\left(0\right)=\bm{H}\in\mathbb{R}^{m,r}$.
Let $\bm{H}=\bm{U}_{H}\bm{S}_{H}\bm{V}_{H}^{T}$ be the compact singular
value decomposition of $\bm{H}$, and let $s_{1},\cdots,s_{r}$ denote
the singular values of $\bm{H}$ in descending order. Then the corresponding
geodesic curve is given by \cite{edelman_optimization_manifolds_1998}
\begin{equation}
\bm{U}\left(t\right)=\left[\bm{U}\bm{V}_{H},\bm{U}_{H}\right]\left[\begin{array}{c}
\cos\bm{S}t\\
\sin\bm{S}t\end{array}\right]\bm{V}_{H}^{T},\label{eq:geodesic-curve}\end{equation}
 where $\cos\bm{S}t\in\mathbb{R}^{r\times r}$ and $\sin\bm{S}t\in\mathbb{R}^{r\times r}$
are $r\times r$ diagonal matrices with diagonal entries $\cos\left(s_{1}t\right),\cdots,\cos\left(s_{r}t\right)$
and $\sin\left(s_{1}t\right),\cdots,\sin\left(s_{r}t\right)$, respectively.

When $\bm{H}$ has rank one, i.e., $s_{2}=s_{3}=\cdots=s_{r}=0$,
the equation for the geodesic curve has a particularly simple form.
In this case, let $\bm{u}_{1},\cdots,\bm{u}_{r}$ be the columns of
the matrix $\bm{U}\bm{V}_{H}$.%
\footnote{Note that $\mbox{span}\left(\bm{U}\right)=\mbox{span}\left(\bm{U}\bm{V}_{H}\right)$.
The starting point (in the Grassmann manifold) does not change.%
} Let $\bm{h}\in\mathcal{U}_{m,1}$ be the left singular vector of
$\bm{H}$ corresponding to the largest singular value. After a change
of variables, the geodesic curve can be written as%
\footnote{Again, although the matrix $\bm{U}\left(t\right)$ in (\ref{eq:geodesic-curve-rank1})
and the matrix $\bm{U}\left(t\right)$ in (\ref{eq:geodesic-curve})
may be different, both matrices generate the same hyperplane in the
Grassmann manifold $\mathcal{G}_{m,r}$. Therefore, Equations (\ref{eq:geodesic-curve})
and (\ref{eq:geodesic-curve-rank1}) describe the same geodesic curve. %
} \begin{equation}
\bm{U}\left(t\right)=\left[\bm{u}_{1}\cos t+\bm{h}\sin t,\bm{u}_{2},\cdots,\bm{u}_{r}\right],\quad t\in\left[0,\pi\right).\label{eq:geodesic-curve-rank1}\end{equation}
 Here, the range of values for the parameter $t$ is restricted to
$\left[0,\pi\right)$, since \begin{align*}
\mbox{span}\left(\bm{U}\left(t+\pi\right)\right) & =\mbox{span}\left(\left[-\bm{u}_{1}\cos t-\bm{h}\sin t,\bm{u}_{2},\cdots,\bm{u}_{r}\right]\right)\\
 & =\mbox{span}\left(\bm{U}\left(t\right)\right),\end{align*}
 and therefore $\mbox{span}\left(\bm{U}\left(t\right)\right)$ is
a periodic function with period $\pi$.

\section{The SET Algorithm - A Two Step Procedure}

\label{sec:SET}

\subsection{The SET algorithm: a high level description}

Our algorithm aims to minimize the objective function $f\left(\bm{U}\right)$.
The basic component is a gradient search approach: for a given estimate
$\bm{U}$, we search in the gradient descent direction for a minimizer.
This part of the algorithm is referred to as {}``subspace evolution''.
The details are presented in Section \ref{sub:Subspace-Evolution}.

The main difficulty that arises during the gradient descent search,
and makes the SET algorithm highly non-trivial, is when one encounters
{}``barriers''. Careful inspection reveals that the objective function
$f$ can be decomposed into a sum of atomic functions, each of which
involves only one column of $\bm{X}_{\Omega}$ (see Section \ref{sub:Subspace-Transfer}
for details). Along the gradient descent path, the individual atomic
functions may imply different search directions: some of the functions
may decrease and some others may increase in the same direction. The
increases of some atomic functions may result in {}``bumps'' in
the $f$ curve, which block the search procedure from reaching a global
optima and are therefore referred to as \emph{barriers}. The main
component of the {}``transfer'' part of the SET algorithm is to
identify whether there exist barriers along the gradient descent path.
Detecting barriers is in general a very difficult task, since one
obviously does not know the locations of global minima. Nevertheless,
we observe that barriers can be detected by the existence of atomic
functions with inconsistent descent directions. Such an inconsistence
can be seen as an indicator for the existence of a barrier. When a
barrier is expected, the algorithm {}``transfers'' the current point
of the line search - i.e., its corresponding space - to the other
side of the barrier, and proceeds with the search from that point.
Such a transfer does not overshoot global minima as we enforce consistency
of the steepest descent directions at the points before and after
the transfer. The details of barrier detection and subspace transfer
are presented in Sections \ref{sub:Subspace-Transfer}, \ref{sub:Barrier-example},
\ref{sub:Barrier-detection}, and \ref{sub:compute-tmin-tmax}.

The major steps of the SET algorithm are given in Algorithm \ref{alg:SET}.
Here, we introduce an error tolerance parameter $\epsilon_{e}>0$.
The stopping criterion is given by $\left\Vert \bm{X}_{\Omega}-\mathcal{P}_{\Omega}\left(\bm{X}^{\prime}\right)\right\Vert _{F}^{2}\le\epsilon_{e}\left\Vert \bm{X}_{\Omega}\right\Vert _{F}^{2}$
where $\bm{X}^{\prime}$ denotes the estimated low-rank matrix. In
our simulations, we set $\epsilon_{e}=10^{-6}$. The SET algorithm
described below only searches for an optimal column space, represented
by $\bm{U}$. Other modifications are possible, as already pointed
out. For example, to speed up the process, one may alternatively optimize
over $\bm{U}$ and $\bm{V}$ (representing the column and row spaces,
respectively). These extensions are not described in the manuscript.

\begin{algorithm}[tbh]

\caption{\label{alg:SET}The SET algorithm}

\textbf{Input}: $\bm{X}_{\Omega}$, $\Omega$, $r$ and $\epsilon_{e}$.

\textbf{Output}: $\bm{X}^{\prime}$.

\textbf{Initialization}: Randomly generate a $\bm{U}\in\mathcal{U}_{m,r}$.

\textbf{Steps}: Execute the following steps iteratively: 
\begin{enumerate}
\item Perform subspace transfer algorithm described in Algorithm \ref{alg:ST}. 
\item Perform subspace evolution algorithm described in Algorithm \ref{alg:SE}. 
\item According to (\ref{eq:f_U}) find the optimal $\bm{W}_{U}$ and set
$\bm{X}^{\prime}=\bm{U}\bm{W}_{U}$. If $\left\Vert \bm{X}_{\Omega}-\mathcal{P}_{\Omega}\left(\bm{X}^{\prime}\right)\right\Vert _{F}^{2}\le\epsilon_{e}\left\Vert \bm{X}_{\Omega}\right\Vert _{F}^{2}$,
output $\bm{X}^{\prime}$ and quit. Otherwise, go to Step 1). 
\end{enumerate}

\end{algorithm}

\subsection{\label{sub:Subspace-Evolution}Subspace evolution}

For the optimization problem at hand, we refine the current column
space estimate $\bm{U}$ using a gradient descent method. For a given
$\bm{U}\in\mathcal{U}_{m,r}$, it is straightforward to solve the
least square problem\begin{equation}
\underset{\bm{W}\in\mathbb{R}^{r\times n}}{\min}\left\Vert \bm{X}_{\Omega}-\mathcal{P}_{\Omega}\left(\bm{U}\bm{W}\right)\right\Vert _{F}^{2}.\label{eq:least-square-W}\end{equation}
 Denote the optimal solution by $\bm{W}_{U}$. Let $\bm{X}_{r}=\bm{X}_{\Omega}-\mathcal{P}_{\Omega}\left(\bm{U}\bm{W}_{U}\right)$
be the residual matrix. Then the gradient%
\footnote{The gradient is well defined almost everywhere in $\mathcal{U}_{m,r}$.%
} of $f$ at $\bm{U}$ is given by \begin{align}
\nabla_{\bm{U}}f & =-2\bm{X}_{r}\bm{W}_{U}^{T}.\label{eq:gradient}\end{align}
 The proof of this claim is given in Appendix \ref{sub:pf-gradient}.
The gradient $\nabla_{\bm{U}}f$ gives the direction along which the
objective function $f$ increases the fastest. In classical gradient
descent methods, the search path direction is opposite to the gradient,
i.e., $-\nabla_{\bm{U}}f$. In order to make the search step more
suitable for the transfer step, we choose the search direction as
follows. Consider the singular value decomposition of the matrix $\nabla_{\bm{U}}f$.
Let $\bm{h}\in\mathcal{U}_{m,1}$ and $\bm{v}\in\mathcal{U}_{r,1}$
be the left and right singular vectors corresponding to the largest
singular value of $\nabla_{\bm{U}}f$.%
\footnote{With probability one, the largest singular value is strictly positive
and distinct from other singular values.%
} Then the search direction is defined as \begin{equation}
\bm{H}=-\bm{h}\bm{v}^{T}.\label{eq:Search-Direction}\end{equation}
 It can be easily verified that if $\nabla_{\bm{U}}f\ne\bm{0}$ then
$\left\langle \bm{H},\nabla_{\bm{U}}f\right\rangle =\mbox{trace}\left(\bm{H}^{T}\nabla_{\bm{U}}f\right)<0$,
and therefore the objective function decreases along the direction
of $\bm{H}$. The geodesic curve starting from $\bm{U}$ and pointing
along $\bm{H}$ can be computed via (\ref{eq:geodesic-curve-rank1}).

The subspace evolution part is designed to search for a {}``neighboring
minimizer'' of the function $f$ \emph{along the geodesic curve}.
It is an analogue of the line search procedure in Euclidean space.
Its continuous counterpart consists of moving the estimate $\bm{U}$
continuously along the direction $\bm{H}$ until the objective function
stops decreasing. For computer simulations, one has to discretize
the continuous counterpart. Our implementation includes two steps.
Let $t^{*}$ denote the neighboring minimizer along the geodesic curve.
The goal of the first step is to identify an upper bound on $t^{*}$,
denoted by $t_{\max}$. Since $f\left(t\right)$ is periodic with
period $\pi$, $t_{\max}$ is upper bounded by $\pi$. The second
step is devoted to locating the minimizer $t^{*}\in\left[0,t_{\max}\right]$
accurately by iteratively applying the golden section rule \cite{gill_practical_optimization_1982}.
These two steps are described in Algorithm \ref{alg:SE}. The constants
are set to $\epsilon=10^{-9}$, $c_{1}=\left(\sqrt{5}-1\right)/2$,
$c_{2}=c_{1}/\left(1-c_{1}\right)$ and $itN=10$. Note that our discretized
implementation is not optimized with respect to its continuous counterpart,
but is sufficiently accurate in practice.

\begin{algorithm}[tbh]

\caption{\label{alg:SE}Subspace evolution.}

\textbf{Input}: $\bm{X}_{\Omega}$, $\Omega$, $\bm{U}$, and $itN$.

\textbf{Output}: $t^{*}$ and $\bm{U}\left(t^{*}\right)$.

\textbf{Initialization}: Compute the gradient and the search direction
according to (\ref{eq:gradient}) and (\ref{eq:Search-Direction})
respectively. The geodesic curve $\bm{U}\left(t\right)$ along the
search direction can be computed via (\ref{eq:geodesic-curve-rank1}).

\textbf{Step A}: find $t_{\max}\le\pi$ such that $t^{*}\in\left[0,t_{\max}\right]$

Let $t^{\prime}=\epsilon\pi$. 
\begin{enumerate}
\item Let $t^{\prime\prime}=c_{2}\cdot t^{\prime}$. If $t^{\prime\prime}>\pi$,
then $t_{\max}=\pi$. Quit Step A. 
\item If $f\left(\bm{U}\left(t^{\prime\prime}\right)\right)>f\left(\bm{U}\left(t\right)\right)$,
then $t_{\max}=t^{\prime\prime}$. Quit Step A. 
\item Otherwise, $t^{\prime}=t^{\prime\prime}$. Go back to step 1). 
\end{enumerate}
\textbf{Step B}: numerically search for $t^{*}$ in $\left[0,t_{\max}\right]$.

Let $t_{1}=t_{\max}/c_{2}^{2}$, $t_{2}=t_{\max}/c_{2}$, $t_{4}=t_{\max}$,
and $t_{3}=t_{1}+c_{1}\left(t_{4}-t_{1}\right)$. Let $itn=1$. Perform
the following iterations. 
\begin{enumerate}
\item If $f\left(\bm{U}\left(t_{1}\right)\right)>f\left(\bm{U}\left(t_{2}\right)\right)>f\left(\bm{U}\left(t_{3}\right)\right)$,
then $t_{1}=t_{2}$, $t_{2}=t_{3}$, and $t_{3}=t_{1}+c_{1}\left(t_{4}-t_{1}\right)$. 
\item Else, $t_{4}=t_{3}$, $t_{3}=t_{2}$ and $t_{2}=t_{1}+\left(1-c_{1}\right)\left(t_{4}-t_{1}\right)$. 
\item $itn=itn+1$. If $itn>itN$, quit the iterations. Otherwise, go back
to step 1). 
\end{enumerate}
Let $t^{*}=\underset{t\in\left\{ t_{1},\cdots,t_{4}\right\} }{\arg\min}f\left(\bm{U}\left(t\right)\right)$
and compute $\bm{U}\left(t^{*}\right)$. 
\end{algorithm}

\subsection{\label{sub:Subspace-Transfer}Subspace transfer}

Unfortunately, the objective function $f\left(\bm{U}\right)$ is typically
not a convex function of $\bm{U}$. The described linear search procedure
may not converge to a global minimum because the search path may be
blocked by what we call {}``barriers''. In subsequent subsections,
we show how {}``barriers'' arise in matrix completion problems and
how to overcome the problem introduced by barriers.

At this point, we formally introduce the decoupling principle. This
principle is essential in understanding the behavior of the objective
function. It implies that the objective function $f\left(\bm{U}\left(t\right)\right)$
can be decoupled into a sum of atomic functions, each of which is
relatively simple to analyze. Specifically, the objective function
$f\left(\bm{U}\left(t\right)\right)$ is the squared Frobenius norm
of the residue matrix; it can be decomposed into a sum of the squared
Frobenius norms of the residue columns. Let $\bm{x}_{\Omega_{j}}\in\mathbb{R}^{m\times1}$
be the $j^{th}$ column of the matrix $\bm{X}_{\Omega}$. Let $\mathcal{P}_{\Omega_{j}}$
be the projection operator corresponding to the $j^{th}$ column,
defined by \begin{align}
\mathcal{P}_{\Omega_{j}}:\;\mathbb{R}^{m} & \rightarrow\mathbb{R}^{m}\nonumber \\
\mathcal{P}_{\Omega_{j}}(\bm{v}) & \mapsto\bm{v}_{\Omega_{j}},\;\mbox{where}\;\left(\bm{v}_{\Omega_{j}}\right)_{i}=\begin{cases}
\bm{v}_{i} & \mbox{if }\left(i,j\right)\in\Omega\\
0 & \mbox{if }\left(i,j\right)\notin\Omega\end{cases}.\label{eq:proj-vector}\end{align}
 Then the objective function $f\left(\bm{u}\left(t\right)\right)$
can be written as a sum of $n$ atomic functions: \begin{align}
f\left(\bm{U}\left(t\right)\right) & =\underset{\bm{W}\in\mathbb{R}^{r\times n}}{\min}\left\Vert \bm{X}_{\Omega}-\mathcal{P}_{\Omega}\left(\bm{U}\left(t\right)\bm{W}\right)\right\Vert _{F}^{2}\nonumber \\
 & =\sum_{j=1}^{n}\underbrace{\min_{\bm{W}_{:j}\in\mathbb{R}^{r}}\left\Vert \bm{x}_{\Omega_{j}}-\mathcal{P}_{\Omega_{j}}\left(\bm{U}\left(t\right)\bm{W}_{:j}\right)\right\Vert _{F}^{2}}_{f_{j}\left(\bm{U}\left(t\right)\right)},\label{eq:f-decoupling}\end{align}
 where $\bm{W}_{:j}$ is the $j^{th}$ column of the matrix $\bm{W}$.
This decoupling principle can be easily verified by the additivity
of the squared Frobenius norm. A formal proof is presented in Appendix
\ref{sub:pf-decoupling}.

We study atomic functions along the geodesic curve in a rank-one direction
(\ref{eq:geodesic-curve-rank1}) and summarize their typical behavior
in the following proposition. 
\begin{prop}
\label{pro:atomic-f-periodic}Let $\bm{U}\left(t\right)$ be of the
form in (\ref{eq:geodesic-curve-rank1}). Given a vector $\bm{x}\in\mathbb{R}^{m}$
and an index set $\Omega\subset\left[m\right]$, consider the function
\begin{equation}
f_{\bm{x},\Omega}\left(\bm{U}\left(t\right)\right)=\underset{\bm{w}\in\mathbb{R}^{r}}{\min}\;\left\Vert \bm{x}_{\Omega}-\mathcal{P}_{\Omega}\left(\bm{U}\left(t\right)\bm{w}\right)\right\Vert _{F}^{2}.\label{eq:proposition-atomic-f-periodic}\end{equation}
 Then either one of the following two claims holds. 
\begin{enumerate}
\item The function $f_{\bm{x},\Omega}\left(\bm{U}\left(t\right)\right)$
is a constant function. 
\item The function $f_{\bm{x},\Omega}\left(\bm{U}\left(t\right)\right)$
is periodic, with period $\pi$. It has a unique minimizer, $t_{\min}\in\left[0,\pi\right)$,
and a unique maximizer, $t_{\max}\in\left[0,\pi\right)$. 
\end{enumerate}
\end{prop}
\vspace{0.01in}

The proof is given in Appendix \ref{sub:pf-atomic-f-periodic} and
the computations of $t_{\min}$ and $t_{\max}$ are detailed in Section
\ref{sub:compute-tmin-tmax}.

\subsection{\label{sub:Barrier-example}Barrier - an illustration}

We use the following example to illustrate the concept of a barrier.
Consider an incomplete observation of a rank-one matrix \[
\bm{X}_{\Omega}=\left[\begin{array}{ccc}
? & 2 & 1\\
3 & ? & 1\\
3 & 2 & ?\end{array}\right],\]
 where question marks denote that the corresponding entries are unknown.
It is clear that the objective function $f\left(\bm{U}\left(t\right)\right)$
is minimized by $\bm{U}_{\bm{X}}=\frac{1}{\sqrt{3}}\left[1,1,1\right]^{T}$,
i.e., $f\left(\bm{U}_{\bm{X}}\right)=0$ and the recovered matrix
equals $\hat{\bm{X}}=\left[1,1,1\right]^{T}\cdot\left[3,2,1\right]$.
Let us study one of the atomic functions, say $f_{1}\left(\bm{U}\right)$.
For any $\bm{U}\in\mathcal{U}_{3,1}$ of the form $\left[\sqrt{1-2\epsilon^{2}},\epsilon,\epsilon\right]^{T}$
with $\epsilon\in\left[-1/\sqrt{2},1/\sqrt{2}\right]\backslash\left\{ 0\right\} $,
one has \[
f_{1}\left(\bm{U}\right)=\min_{w\in\mathbb{R}}\;\left\Vert \left[\begin{array}{c}
0\\
3\\
3\end{array}\right]-\left[\begin{array}{c}
0\\
\epsilon\\
\epsilon\end{array}\right]w\right\Vert _{F}^{2}=0.\]
 Similarly, For any $\bm{U}$ of the form $\left[\sqrt{1-2\epsilon^{2}},\epsilon,-\epsilon\right]^{T}$
with $\epsilon\in\left[-1/\sqrt{2},1/\sqrt{2}\right]$, one has \[
f_{1}\left(\bm{U}\right)=\min_{w\in\mathbb{R}}\;\left\Vert \left[\begin{array}{c}
0\\
3\\
3\end{array}\right]-\left[\begin{array}{c}
0\\
\epsilon\\
-\epsilon\end{array}\right]w\right\Vert _{F}^{2}=18.\]
 As a result, \begin{align*}
 & f_{1}\left(\bm{U}\right)=0,\quad\mbox{if }\bm{U}_{2}=\bm{U}_{3}\ne0;\\
 & f_{1}\left(\bm{U}\right)=18,\quad\mbox{if }\bm{U}_{2}=-\bm{U}_{3}.\end{align*}
 This gives us the two contours depicted in Fig. \ref{fig:contour}
(projected on the plane spanned by $\bm{U}_{2}$ and $\bm{U}_{3}$,
the second and the third entries of the vector $\bm{U}$ respectively).
Suppose that one starts with the initial guess $\bm{U}\left(0\right)=\frac{1}{\sqrt{102}}\left[-10,1,1\right]^{T}$.
Then $f\left(\bm{U}\left(0\right)\right)=\sum_{i=1}^{3}f_{i}\left(\bm{U}\left(0\right)\right)\le0+8+2=10$.
On the other hand, for any $\bm{U}$ in the preimage of $f_{1}\left(\bm{U}\right)=18$,
one has $f\left(\bm{U}\right)\ge18>10\ge f\left(\bm{U}\left(0\right)\right)$.
As a result, any gradient descent method (continuous version) can
not lead the estimate $\bm{U}\left(t\right)$ to cross the contour
$\left\{ \bm{U}:\; f_{1}\left(\bm{U}\right)=18\right\} $. That is,
the contour $f_{1}=18$ forms a {}``barrier'' for the line search
procedure. A more careful analysis reveals that the objective function
$f$ is not continuous at the point $\bm{U}=\left[1,0,0\right]^{T}$.
Our extensive simulations suggest that a gradient descent procedure
is typically trapped towards these singular points. See Fig. \ref{fig:search-path}
for an illustration of this phenomenon.

\begin{figure}
\begin{centering}
\hfill{}\subfloat[\label{fig:contour}Contours of $f_{1}$.]{\begin{centering}
\includegraphics[scale=0.75]{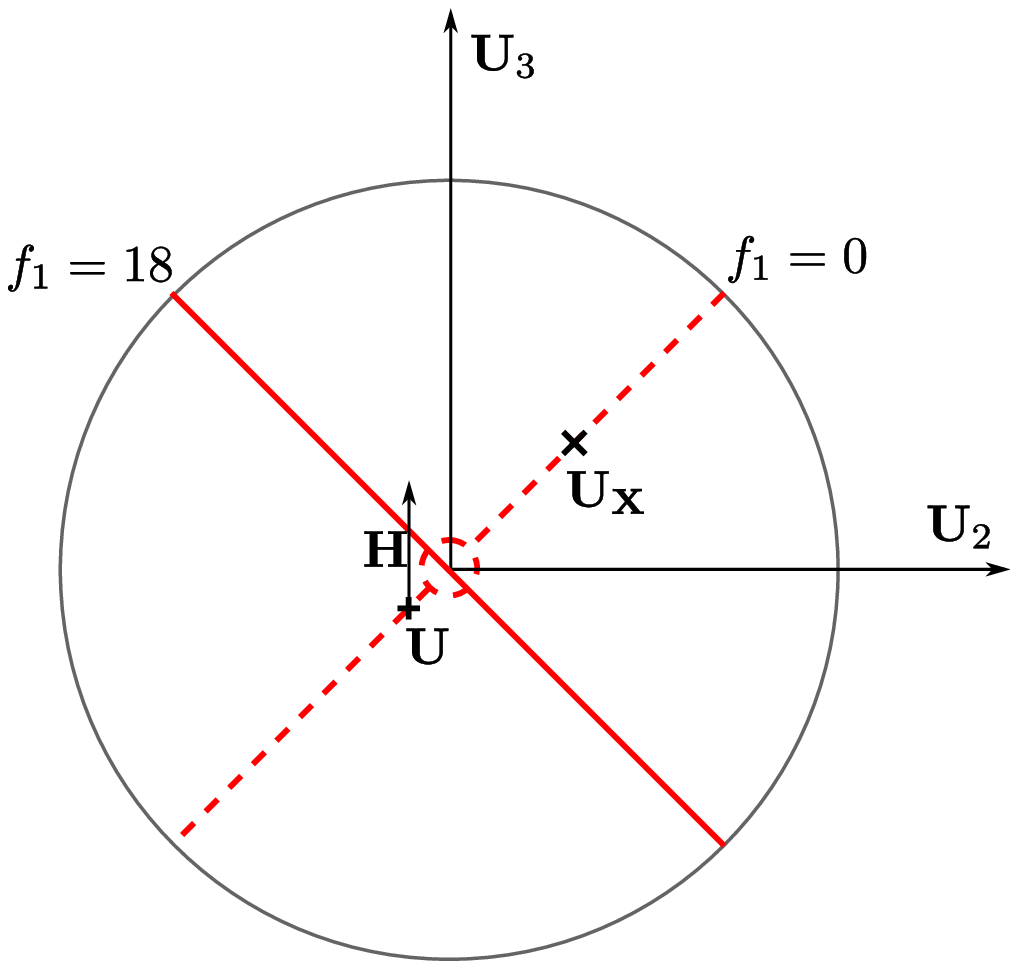}
\par\end{centering}

}\hfill{}\subfloat[\label{fig:search-path}Search paths with zooming in.]{\begin{centering}
\includegraphics[scale=0.75]{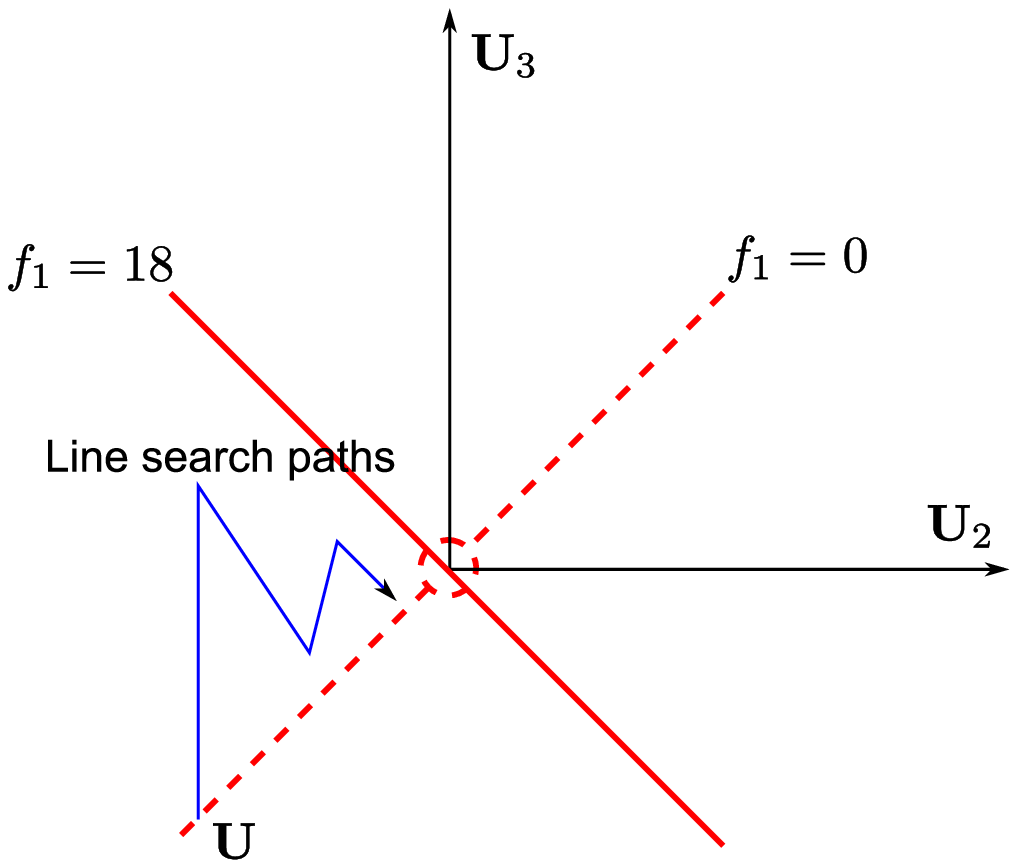}
\par\end{centering}

}\hfill{}
\par\end{centering}

\caption{\label{fig:example}An illustrative example for barriers.}

\end{figure}

\subsection{\label{sub:Barrier-detection}Barrier Detection and Subspace Transfer}

We describe a heuristic procedure for detecting barriers and transferring
the current estimate $\bm{U}$ from one side of a barrier to the other
side. 

The intuition behind barrier detection is as follows. Recall that
every atomic function is periodic and has a unique minimizer and maximizer
in one period. In the gradient descent direction, some atomic function
increase while some others decrease. On the other hand, in the matrix
completion problem, the objective function reaches zero at a global
minimizer. This implies that each atomic function reaches its minimum
at a global minimizer. That is, in a small neighborhood of a global
minimizer, the atomic functions should be {}``consistent'': there
should exist a small $\epsilon>0$ such that when current estimate
$\bm{U}$ is $\epsilon$-close to the global minimizer $\bm{U}_{\bm{X}}$,
there is no atomic function reaching its maximum value along the path
from current estimate $\bm{U}$ to the global minimizer $\bm{U}_{\bm{X}}$.
Following this intuition, we have the following definition of barriers.
Consider the geodesic path in (\ref{eq:geodesic-curve-rank1}) starting
from $\bm{U}$, pointing in the direction $\bm{H}$. Denote the unique
minimizer and maximizer of the $k^{th}$ atomic function by $t_{\min,k}$
and $t_{\max,k}$ (for constant atomic functions, we set $t_{\min,k}=t_{\max,k}=0$).
Refer to the atomic functions that decrease in the direction of $\bm{H}$
as consistent atomic functions. We say that the maximizer of the $k^{th}$
atomic function forms a barrier if 
\begin{enumerate}
\item In the $\bm{H}$ direction, there exists a consistent atomic function,
say the $j^{th}$ atomic function, such that the maximizer of the
$k^{th}$ atomic function appears before the minimizer of the $j^{th}$
atomic function. That is, there exists $j\in\left[n\right]$ such
that $0<t_{\max,k}<t_{\min,j}<t_{\max,j}<\pi$. 
\item The gradients of $f$ at $\bm{U}\left(0\right)$ and $\bm{U}\left(t_{\max,k}\right)$
are consistent (form a sharp angle), i.e., $\frac{d}{dt}f\left(\bm{U}\left(t\right)\right)|_{t=t_{\max,k}}<0$.
In Appendix \ref{sub:consistency}, we describe how to decide whether
$\frac{d}{dt}f\left(\bm{U}\left(t\right)\right)|_{t=t_{\max,k}}<0$. 
\end{enumerate}
Moreover, we say that the $j^{th}$ column of $\bm{X}_{\Omega}$ admits
barriers if there exists a $k\in\left[n\right]$ such that the maximizer
of the $k^{th}$ atomic function forms a barrier and $t_{\max,k}<t_{\min,j}<t_{\max,j}$.

\vspace{0cm}

Once barriers are detected, we transfer $\bm{U}$. To avoid overshooting,
the transfer destination should be {}``$\epsilon$-close'' to the
barrier. As $\epsilon\rightarrow0$, the transfer destination is on
the barrier ($\bm{U}\left(t_{\max,k}\right)$ for some $k$). In our
implementation, we focus on the {}``closest'' barriers to $\bm{U}$.
Define\begin{equation}
\mathcal{J}=\left\{ j:\;\mbox{the }j^{th}\mbox{ column of }\bm{X}_{\Omega}\;\mbox{admits barriers}\right\} ,\label{eq:def-set-J}\end{equation}
\begin{equation}
j^{*}=\underset{j\in\mathcal{J}}{\arg\min}\; t_{\min,j},\;\mbox{and}\label{eq:def-j-star}\end{equation}
\begin{align}
k^{*} & =\underset{k}{\arg\max}\;\left\{ t_{\max,k}:\;\mbox{the maximizer of the}\; k^{th}\;\mbox{atomic}\right.\nonumber \\
 & \qquad\left.\mbox{function forms a barrier and }t_{\max,k}<t_{\min,j^{*}}\right\} .\label{eq:def-k-star}\end{align}
 We transfer our current estimation $\bm{U}\left(0\right)$ to $\bm{U}\left(t_{\max,k^{*}}\right)$.

The subspace transfer part is a combination of barrier detection and
column space transfer. It is described in Algorithm \ref{alg:ST}.

\begin{algorithm}[tbh]

\caption{\label{alg:ST}Subspace transfer}

\textbf{Input}: $\bm{X}_{\Omega}$, $\Omega$, and $\bm{U}$.

\textbf{Output}: $t_{\mbox{tran}}$ and $\bm{U}\left(t_{\mbox{tran}}\right)$.

\textbf{Steps}: 
\begin{enumerate}
\item Compute $t_{\max,j}$ and $t_{\min,j}$ for each column $j$. 
\item Check whether there exist barriers.

\begin{enumerate}
\item Find $j^{*}$ and $k^{*}$ according to (\ref{eq:def-j-star}) and
(\ref{eq:def-k-star}), respectively. 
\item Let $t_{\mathrm{tran}}=t_{\max,k^{*}}$ and compute $\bm{U}\left(t_{\mathrm{tran}}\right)$
according to (\ref{eq:geodesic-curve-rank1}). 
\end{enumerate}
\item If no barrier is detected (the set $\mathcal{J}$ in (\ref{eq:def-set-J})
is empty), then $t_{\mbox{tran}}=0$ and $\bm{U}\left(t_{\mbox{tran}}\right)=\bm{U}$. 
\end{enumerate}

\end{algorithm}

\vspace{0.01in}

\subsection{\label{sub:compute-tmin-tmax}Computation of $t_{\min}$ and $t_{\max}$}

The subspace transfer part of the SET algorithm relies on the minimizers
and maximizers of atomic functions. This subsection presents the details
for computing these extremals. 

Let $\bm{U}\left(t\right)$ be of the form in (\ref{eq:geodesic-curve-rank1}).
Also, let $\Omega\subset\left[m\right]$ be an index set. Define \begin{align*}
\bm{U}_{\Omega}\left(t\right) & =\left[\mathcal{P}_{\Omega}\left(\bm{u}_{1}\cos t+\bm{h}\sin t\right),\mathcal{P}_{\Omega}\left(\bm{u}_{2}\right),\cdots,\mathcal{P}_{\Omega}\left(\bm{u}_{r}\right)\right]\\
 & =\left[\bm{u}_{1,\Omega}\cos t+\bm{h}_{\Omega}\sin t,\bm{u}_{2,\Omega},\cdots,\bm{u}_{r,\Omega}\right].\end{align*}
 For a given vector $\bm{x}\in\mathbb{R}^{m}$, denote $\mathcal{P}_{\Omega}\left(\bm{x}\right)$
by $\bm{x}_{\Omega}$. Define \[
\bm{x}_{\Omega,r}\left(t\right)=\bm{x}_{\Omega}-\mathcal{P}\left(\bm{x}_{\Omega},\bm{U}_{\Omega}\left(t\right)\right).\]
 The above expression simply specifies the projection residue vector
of $\bm{x}_{\Omega}$, where the projection is performed on the hyperplane
$\mbox{span}\left(\bm{U}_{\Omega}\left(t\right)\right)$. Note that
$\bm{x}_{\Omega,r}\left(t\right)$ is a function of $\bm{u}_{2,\Omega},\cdots,\bm{u}_{r,\Omega}$.

We would like to understand how $\bm{x}_{\Omega,r}\left(t\right)$
changes with $t$. Note that $\bm{u}_{2,\Omega},\cdots,\bm{u}_{r,\Omega}$
do not change with $t$. We shall find an expression of $\bm{x}_{\Omega,r}\left(t\right)$
that does not directly include $\bm{u}_{2,\Omega},\cdots,\bm{u}_{r,\Omega}$.
For this purpose, let \begin{align*}
\bm{x}_{r}^{\prime} & =\bm{x}_{\Omega}-\mathcal{P}\left(\bm{x}_{\Omega},\left[\bm{u}_{2,\Omega},\cdots,\bm{u}_{r,\Omega}\right]\right),\\
\bm{u}_{r} & =\bm{u}_{1,\Omega}-\mathcal{P}\left(\bm{u}_{1,\Omega},\left[\bm{u}_{2,\Omega},\cdots,\bm{u}_{r,\Omega}\right]\right),\;\mbox{and}\\
\bm{h}_{r} & =\bm{h}_{\Omega}-\mathcal{P}\left(\bm{h}_{\Omega},\left[\bm{u}_{2,\Omega},\cdots,\bm{u}_{r,\Omega}\right]\right).\end{align*}
 Let \[
\bm{u}_{r}\left(t\right)=\bm{u}_{r}\cos t+\bm{h}_{r}\sin t.\]
According to Proposition \ref{pro:projection-high-rank} in Appendix
\ref{sub:pf-atomic-f-periodic}, we have \[
\bm{x}_{\Omega,r}\left(t\right)=\bm{x}_{r}^{\prime}-\mathcal{P}\left(\bm{x}_{r}^{\prime},\bm{u}_{r}\left(t\right)\right).\]
 Note that $\bm{u}_{r}\left(t\right)$ has a simpler form compared
to $\bm{U}\left(t\right)$, and is therefore easier to analyze.

According to Proposition \ref{pro:atomic-f-periodic}, the function
$f_{\bm{x},\Omega}\left(t\right)=\left\Vert \bm{x}_{\Omega,r}\left(t\right)\right\Vert ^{2}$
is either a constant function or a periodic function with a unique
maximizer and minimizer in one period $\pi$. We are interested in
computing the unique maximizer and minimizer, denoted by $t_{\max}$
and $t_{\min}$ respectively, when the function is not constant. Apply
Proposition \ref{pro:atomic-fn-rank1} in Appendix \ref{sub:pf-atomic-f-periodic},
the following procedure generates the values of $t_{\max}$ and $t_{\min}$. 
\begin{enumerate}
\item Check whether

\begin{enumerate}
\item the vectors $\bm{u}_{r}$ and $\bm{h}_{r}$ are linearly \emph{dependent},
or 
\item the vector $\bm{x}_{r}$ is orthogonal to both $\bm{u}_{r}$ and $\bm{h}_{r}$. 
\end{enumerate}
If either of the above two properties holds, then $f_{\bm{x},\Omega}\left(t\right)$
is a constant function. Set $t_{\min}=t_{\max}=0$ and quit the procedure.

\item Let \[
\bm{c}=\left[\begin{array}{c}
c_{1}\\
c_{2}\end{array}\right]=\left[\bm{u}_{r},\bm{h}_{r}\right]^{\dagger}\bm{x}_{r},\]
 where the superscript $\dagger$, as before, denotes the pseudoinverse.
Define a mapping \begin{align}
\mbox{atan}:\;\mathbb{R}\times\mathbb{R} & \rightarrow\left[0,\pi\right)\nonumber \\
\left(x_{1},x_{2}\right) & \mapsto\begin{cases}
\pi/2\\
\qquad\mbox{if}\; x_{2}=0,\\
\mbox{tan}^{-1}\left(x_{1}/x_{2}\right)\\
\qquad\mbox{if}\; x_{2}\ne0\;\mbox{and}\; x_{1}/x_{2}\ge0,\\
\pi-\mbox{tan}^{-1}\left(-x_{1}/x_{2}\right)\\
\qquad\mbox{if}\; x_{2}\ne0\;\mbox{and}\; x_{1}/x_{2}<0.\end{cases}\label{eq:def-atan}\end{align}
 Then \[
t_{\max}=\mbox{atan}\left(c_{2},c_{1}\right).\]

\item The minimizer $t_{\min}$ is computed via\[
t_{\min}=\mbox{atan}\left(\bm{x}_{r}^{T}\bm{u}_{r},-\bm{x}_{r}^{T}\bm{h}_{r}\right).\]

\end{enumerate}

\section{\label{sec:Simulations}Performance Evaluation}

We tested the SET algorithm by randomly generating low-rank matrices
$\bm{X}$ and index sets $\Omega$. Specifically, we decomposed the
matrix $\bm{X}$ into $\bm{X}=\bm{U}_{\bm{X}}\bm{S}_{\bm{X}}\bm{V}_{\bm{X}}^{T}$,
where $\bm{U}_{\bm{X}}\in\mathcal{U}_{m,r}$, $\bm{V}_{\bm{X}}\in\mathcal{U}_{n,r}$,
and $\bm{S}_{\bm{X}}\in\mathbb{R}^{r\times r}$. We generated $\bm{U}_{\bm{X}}$
and $\bm{V}_{\bm{X}}$ from the isotropic distribution on the set
$\mathcal{U}_{m,r}$ and $\mathcal{U}_{n,r}$, respectively. The entries
of the $\bm{S}_{\bm{X}}$ matrix were independently drawn from the
standard Gaussian distribution $\mathcal{N}\left(0,1\right)$. This
step is important in order to guarantee randomness in the singular
values of $\bm{X}$. The index set $\Omega$ is also randomly generated
according to a uniform distribution over the set $\left\{ \Omega^{\prime}\subset\left[m\right]\times\left[n\right]:\;\left|\Omega^{\prime}\right|=k\right\} $,
for some constant $k$.

The performance of the SET algorithm is excellent, when compared to
the performance of other low-rank completion methods. We tested different
matrices with different ranks and different sampling rates, defined
as $\left|\Omega\right|/\left(m\times n\right)$. Fig. \ref{fig:improvement-transfer-step}
illustrates the performance improvement due to the subspace transfer
step. Significant gain is observed by integrating the subspace evolution
and subspace transfer steps. Fig. \ref{fig:SET-performance} shows
the performance of the SET algorithm for several choices of matrix
sizes and ranks. We also compare the SET algorithm to other matrix
completion algorithms%
\footnote{Though the SVT algorithm is not designed to solve the problem (P0),
we include it for completeness. In the standard SVT algorithm, there
is no explicit constraint on the rank of the \emph{reconstructed}
matrix. For fair comparison, we take the best rank-$r$ approximation
of the reconstructed matrix, and check whether it satisfies the performance
criterion.%
}. As shown in Figure \ref{fig:SET-performance-comparison}, the SET
algorithm outperforms all other tested completion approaches. One
unique property of the SET algorithm is that it works well in both
high sampling rate and low sampling rate regimes: in the high sampling
rate regime, the SET algorithm finds the unique low-rank solution;
in the low sampling rate regime, it finds one of the possibly multiple
low-rank solutions. Also note that there exists a region of sampling
rates for which the SET algorithm (actually all tested algorithms)
exhibits poor performance: the width and critical density of this
region depends on the matrix dimension and rank, and this regions
moves to the right as the rank increases. 

\begin{figure}[tbh]
\begin{centering}
\includegraphics[scale=0.65]{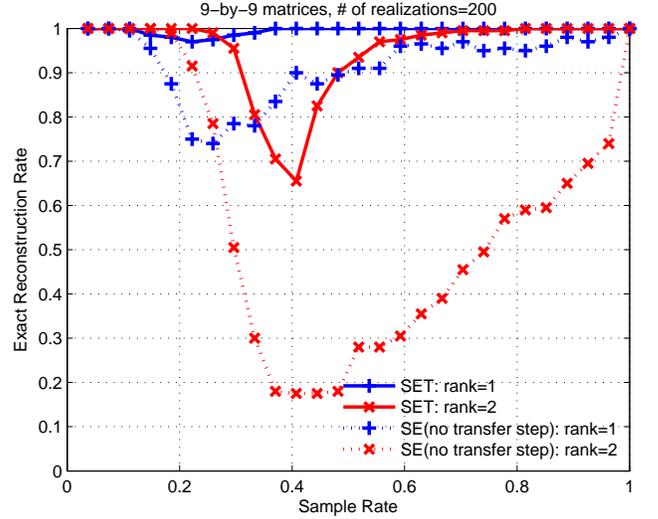}
\par\end{centering}

\caption{\label{fig:improvement-transfer-step}Performance improvement due
to the subspace transfer step.}

\end{figure}

\begin{figure}[tbh]
\begin{centering}
\includegraphics[scale=0.65]{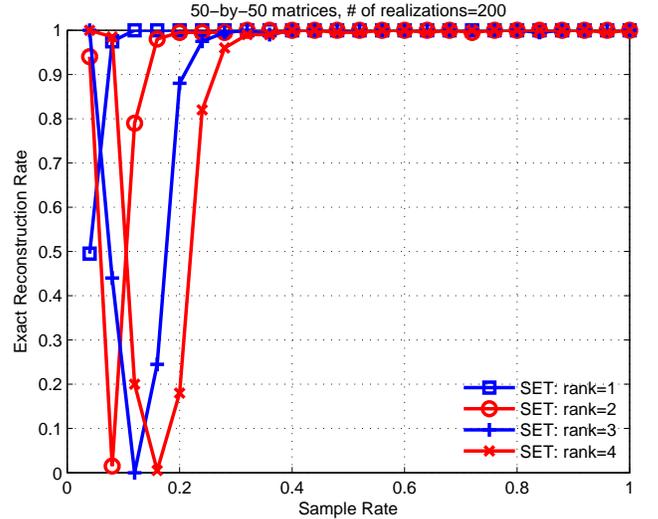}
\par\end{centering}

\caption{\label{fig:SET-performance}Performance of the SET algorithm.}

\end{figure}

\begin{figure}[tbh]
\begin{centering}
\includegraphics[scale=0.65]{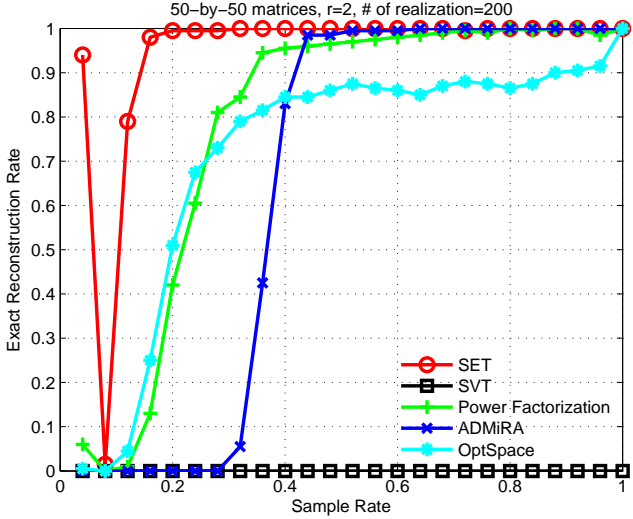} 
\par\end{centering}

\caption{\label{fig:SET-performance-comparison}Performance comparison.}

\end{figure}

Finally, we would like to comment on the complexity of the SET algorithm.
The computational complexity is related to the number of iterations
required for convergence. Since it incorporates a gradient descent
part, the SET algorithm inherits the general disadvantages of a gradient
descent approach: the algorithm may take a large number of iterations
to converge; within each iteration, finding the optimal step size
can be time consuming. Furthermore, extra computations are required
for the subspace transfer step. At the current stage, we do not have
an accurate analytical estimate of the computational complexity. 

\appendix

\subsection{\label{sub:pf-gradient}Proof of the form of the gradient in (\ref{eq:gradient})}

Let $\bm{F}_{\bm{U}}$ be the $m\times r$ matrix of partial derivatives,
i.e., $\left(\bm{F}_{\bm{U}}\right)_{i,j}=\partial f/\partial\bm{U}_{i,j}$.
We first write the objective function via the trace function: \begin{align*}
f & =\left\langle \mathcal{P}_{\Omega}\left(\bm{X}_{\Omega}-\bm{U}\bm{W}_{U}\right),\mathcal{P}_{\Omega}\left(\bm{X}_{\Omega}-\bm{U}\bm{W}_{U}\right)\right\rangle \\
 & \overset{\left(a\right)}{=}\left\langle \bm{X}-\bm{U}\bm{W}_{U},\mathcal{P}_{\Omega}^{*}\left(\mathcal{P}_{\Omega}\left(\bm{X}_{\Omega}-\bm{U}\bm{W}_{U}\right)\right)\right\rangle \\
 & \overset{\left(b\right)}{=}\left\langle \bm{X}-\bm{U}\bm{W}_{U},\mathcal{P}_{\Omega}\left(\bm{X}_{\Omega}-\bm{U}\bm{W}_{U}\right)\right\rangle ,\\
 & =\mbox{trace}\left(\left(\bm{X}-\bm{U}\bm{W}_{U}\right)^{T}\mathcal{P}_{\Omega}\left(\bm{X}_{\Omega}-\bm{U}\bm{W}_{U}\right)\right)\end{align*}
 where the symbol $\mathcal{P}_{\Omega}^{*}$ in $\left(a\right)$
denotes the adjoint operator of $\mathcal{P}_{\Omega}$. Equation
$\left(a\right)$ follows from the definition of the adjoint operator,
and equation $\left(b\right)$ holds because the operator $\mathcal{P}_{\Omega}$
is self-adjoint and idempotent. Note that \[
\frac{\partial f}{\partial\bm{U}_{i,j}}=\left.\frac{\partial f}{\partial\bm{U}_{i,j}}\right|_{\bm{W}_{U}}+\sum_{k,\ell}\left.\frac{\partial f}{\partial\left(\bm{W}_{U}\right)_{k,\ell}}\right|_{\bm{U}}\frac{\partial\left(\bm{W}_{U}\right)_{k,\ell}}{\partial\bm{U}_{i,j}}.\]
 Since $\bm{W}_{U}$ is the solution of the least square problem in
(\ref{eq:least-square-W}), we have \[
\left.\frac{\partial f}{\partial\left(\bm{W}_{U}\right)_{k,\ell}}\right|_{\bm{U}}=0,\;\mbox{for all}\;1\le k\le r\;\mbox{and}\;1\le\ell\le n.\]
 Therefore, \begin{align*}
\bm{F}_{\bm{U}} & =\frac{\partial f}{\partial\bm{U}}=\left.\frac{\partial f}{\partial\bm{U}}\right|_{\bm{W}_{U}}\\
 & =-2\mathcal{P}_{\Omega}\left(\bm{X}_{\Omega}-\bm{U}\bm{W}_{U}\right)\bm{W}_{U}^{T}=-2\bm{X}_{r}\bm{W}_{U}^{T}.\end{align*}

According to \cite[pg. 20]{edelman_optimization_manifolds_1998},
the corresponding tangent vector $\nabla_{\bm{U}}f$ (with respect
to the Grassmann manifold) is given by $\nabla_{\bm{U}}f=\bm{F}_{\bm{U}}-\bm{U}\bm{U}^{T}\bm{F}_{\bm{U}}.$
Since $\bm{W}_{U}$ minimizes the Frobenius norm, it is straightforward
to verify that $\bm{U}$ is orthogonal to $\bm{X}_{r}$, i.e., $\bm{U}^{T}\bm{X}_{r}=\bm{0}$.
Therefore, $\nabla_{\bm{U}}f=\bm{F}_{\bm{U}}=-2\bm{X}_{r}\bm{W}_{U}^{T}$
which proves (\ref{eq:gradient}).

\subsection{\label{sub:pf-decoupling}Proof of the decoupling principle in (\ref{eq:f-decoupling})}

Arbitrarily pick a $\bm{U}\in\mathbb{R}^{m\times r}$. For the matrix
$\bm{X}_{\Omega}$, the objective function $\left\Vert \bm{X}_{\Omega}-\mathcal{P}_{\Omega}\left(\bm{U}\bm{W}\right)\right\Vert _{F}^{2}$
is convex in $\bm{W}$. Let $\bm{W}^{\left(0\right)}$ be a global
minimizer for this function. For each column of $\bm{X}_{\Omega}$,
say $\bm{x}_{\Omega_{j}}$, the function $\left\Vert \bm{x}_{\Omega_{j}}-\mathcal{P}_{\Omega_{j}}\left(\bm{U}\bm{W}_{:,j}\right)\right\Vert ^{2}$
is also convex. Let $\bm{W}_{:j}^{\left(1\right)}$ now be the global
minimizer for this $j^{th}$ atomic function. Concatenate $\bm{W}_{:1}^{\left(1\right)},\cdots,\bm{W}_{:n_{2}}^{\left(1\right)}$
into a matrix and denote the resulting matrix by $\bm{W}^{\left(1\right)}$.
By the additivity of the squared Frobenius norm, the right side of
(\ref{eq:f-decoupling}) becomes $\left\Vert \bm{X}_{\Omega}-\mathcal{P}_{\Omega}\left(\bm{U}\bm{W}^{\left(1\right)}\right)\right\Vert _{F}^{2}$.
By the definition of $\bm{W}^{\left(0\right)}$, $\left\Vert \bm{X}_{\Omega}-\mathcal{P}_{\Omega}\left(\bm{U}\bm{W}^{\left(0\right)}\right)\right\Vert _{F}^{2}\le\left\Vert \bm{X}_{\Omega}-\mathcal{P}_{\Omega}\left(\bm{U}\bm{W}^{\left(1\right)}\right)\right\Vert _{F}^{2}$.
On the other hand, \begin{align*}
\left\Vert \bm{X}_{\Omega}-\mathcal{P}_{\Omega}\left(\bm{U}\bm{W}^{\left(1\right)}\right)\right\Vert _{F}^{2} & =\sum_{j=1}^{n_{2}}\left\Vert \bm{x}_{\Omega_{j}}-\mathcal{P}_{\Omega_{j}}\left(\bm{U}\bm{W}_{:j}^{\left(1\right)}\right)\right\Vert _{F}^{2}\\
 & \le\sum_{j=1}^{n_{2}}\left\Vert \bm{x}_{\Omega_{j}}-\mathcal{P}_{\Omega_{j}}\left(\bm{U}\bm{W}_{:j}^{\left(0\right)}\right)\right\Vert _{F}^{2}\\
 & =\left\Vert \bm{X}_{\Omega}-\mathcal{P}_{\Omega}\left(\bm{U}\bm{W}^{\left(0\right)}\right)\right\Vert _{F}^{2}.\end{align*}
 This proves equation (\ref{eq:f-decoupling}).

\subsection{\label{sub:consistency}Determination of Consistency}

Let $\bm{G}=\left.\nabla_{\bm{U}}f\right|_{\bm{U}\left(t_{\max,k}\right)}$
be the gradient of $f$ at $\bm{U}\left(t_{\max,k}\right)$. It can
be computed via (\ref{eq:gradient}). Consider the geodesic curve
in (\ref{eq:geodesic-curve-rank1}). Define \[
\bm{H}\left(t\right)=\left[-\bm{u}_{1}\sin t+\bm{h}\cos t,\bm{0},\cdots,\bm{0}\right]\;\mbox{for}\; t\in\left[0,\pi\right).\]
 It can be shown that $\bm{H}\left(t_{\max,k}\right)$ is the parallel
transportation of $\bm{H}$ at $t_{\max,k}$ (see \cite[pg. 19]{edelman_optimization_manifolds_1998}
for more details). Based on the definition of the gradient, it can
be shown that $\frac{d}{dt}f\left(\bm{U}\left(t\right)\right)<0$
if and only if \[
\left\langle \bm{G},\bm{H}\left(t_{\max,k}\right)\right\rangle =\bm{G}^{T}\bm{H}\left(t_{\max,k}\right)<0.\]

\subsection{\label{sub:pf-atomic-f-periodic}Proof of Proposition \ref{pro:atomic-f-periodic}}

This subsection presents the proof of Proposition \ref{pro:atomic-f-periodic}
and the mechanism in Section \ref{sub:compute-tmin-tmax} for computing
$t_{\max}$ and $t_{\min}$. We first study the case $r=1$ and then
extend the results to the general case where $r>1$.

In the rank-one case, the geodesic curve has the form $\bm{U}\left(t\right)=\bm{u}\cos t+\bm{h}\sin t$,
with $t\in\left[0,\pi\right)$. For some $\Omega\subset\left[m\right]$,
an atomic function can be written as $\left\Vert \bm{x}_{\Omega}-\mathcal{P}\left(\bm{x}_{\Omega},\bm{u}_{\Omega}\cos t+\bm{h}_{\Omega}\sin t\right)\right\Vert ^{2}$,
where $\bm{u}_{\Omega}=\mathcal{P}_{\Omega}\left(\bm{u}\right)$ and
$\bm{h}_{\Omega}=\mathcal{P}_{\Omega}\left(\bm{h}\right)$. Note that
$\bm{u}_{\Omega}$ may not be of unit norm. For notational convenience,
we drop the subscript $\Omega$. The following proposition describes
the general behavior of an atomic function. 
\begin{prop}
\label{pro:atomic-fn-rank1}Let $\bm{y},\bm{u}_{1},\bm{u}_{2}\in\mathbb{R}^{m}$.
Suppose that 
\begin{enumerate}
\item The vectors $\bm{u}_{1}$ and $\bm{u}_{2}$ are linearly independent. 
\item The vector $\bm{y}$ is not orthogonal to both $\bm{u}_{1}$ and $\bm{u}_{2}$
simultaneously. 
\end{enumerate}
Let $\bm{u}\left(t\right)=\bm{u}_{1}\cos t+\bm{u}_{2}\sin t$ where
$t\in\mathbb{R}$. Define $\bm{y}_{r}\left(t\right)=\bm{y}-\mathcal{P}\left(\bm{y},\bm{u}\left(t\right)\right)$
and $f\left(t\right)=\left\Vert \bm{y}_{r}\left(t\right)\right\Vert ^{2}$.
Then the following is true. 
\begin{enumerate}
\item $f\left(t\right)$ is a periodic function with period $\pi$. 
\item \label{enu:tmin_tmax}$f\left(t\right)$ has a unique minimizer $t_{\min}$
and a unique maximizer $t_{\max}$. 
\item The maximizer $t_{\max}$ defined in \ref{enu:tmin_tmax}) can be
computed in the following way. Let $\bm{c}=\left[c_{1},c_{2}\right]^{T}=\mbox{coeff}\left(\bm{y},\left[\bm{u}_{1},\bm{u}_{2}\right]\right)$.
Then $t_{\max}=\mbox{atan}\left(c_{2},c_{1}\right)$, where the atan
function is defined in (\ref{eq:def-atan}). 
\item The minimizer $t_{\min}$ defined in \ref{enu:tmin_tmax}) is computed
via $t_{\min}=\mbox{atan}\left(\bm{y}^{T}\bm{u}_{i},-\bm{y}^{T}\bm{u}_{2}\right)$. 
\end{enumerate}
\end{prop}
\begin{proof}
This first part is proved by observing that $\bm{u}\left(t+\pi\right)=-\bm{u}\left(t\right)$.
Note that for a given $t$, \[
\bm{y}_{r}\left(t\right)=\bm{y}-\left(\bm{y}^{T}\bm{u}\left(t\right)/\left\Vert \bm{u}\left(t\right)\right\Vert ^{2}\right)\bm{u}\left(t\right).\]
 One has \begin{align*}
\bm{y}_{r}\left(t+\pi\right) & =\bm{y}-\left(\bm{y}^{T}\bm{u}\left(t+\pi\right)/\left\Vert \bm{u}\left(t+\pi\right)\right\Vert ^{2}\right)\bm{u}\left(t+\pi\right)\\
 & =\bm{y}-\left(-\bm{y}^{T}\bm{u}\left(t\right)/\left\Vert \bm{u}\left(t\right)\right\Vert ^{2}\right)\left(-\bm{u}\left(t\right)\right)\\
 & =\bm{y}_{r}\left(t\right).\end{align*}

The other claims of this proposition are proved as follows. By assumption,
$\bm{u}_{1}$ and $\bm{u}_{2}$ are linearly independent. As a result,
$\mbox{span}\left(\left[\bm{u}_{1},\bm{u}_{2}\right]\right)$ is a
hyperplane with dimension two. It is clear that $\bm{u}\left(t\right)=\bm{u}_{1}\cos t+\bm{u}_{2}\sin t\ne0$
for all $t\in\mathbb{R}$ and it forms an ellipse on the hyperplane
$\mbox{span}\left(\left[\bm{u}_{1},\bm{u}_{2}\right]\right)$ centered
at 0. Any line in the hyperplane $\mbox{span}\left(\left[\bm{u}_{1},\bm{u}_{2}\right]\right)$
through the origin can be uniquely represented by a point on the half
ellipse $\bm{u}\left(t\right)$ with $t\in\left[0,\pi\right)$: that
is, for all unit vector $\bm{u}^{\prime}\in\mbox{span}\left(\left[\bm{u}_{1},\bm{u}_{2}\right]\right)$,
there exists a unique $t\in\left[0,\pi\right)$ and an $s\in\mathbb{R}$
such that $\bm{u}^{\prime}s=\bm{u}\left(t\right)$. In other words,
the half ellipse $\bm{u}\left(t\right)$ with $t\in\left[0,\pi\right)$
presents all possible lines (through the origin) in the hyperplane
$\mbox{span}\left(\left[\bm{u}_{1},\bm{u}_{2}\right]\right)$.

Let $\bm{y}_{p}$ be the projection of $\bm{y}$ on the hyperplane
$\mbox{span}\left(\left[\bm{u}_{1},\bm{u}_{2}\right]\right)$, i.e.,
$\bm{y}_{p}=\mbox{proj}\left(\bm{y},\left[\bm{u}_{1},\bm{u}_{2}\right]\right)$.
It is clear that $f\left(t\right)$ is maximized when $\bm{u}\left(t\right)$
is aligned with $\bm{y}_{p}$: this means, there exists a constant
$c\in\mathbb{R}$ such that $\bm{u}\left(t\right)=c\bm{y}_{p}$. By
the definition of the projection, we have $\bm{y}_{p}=\left[\bm{u}_{1},\bm{u}_{2}\right]\bm{c}=\bm{u}_{1}c_{1}+\bm{u}_{2}c_{2}$.
Therefore, $t_{\max}=\mbox{atan}\left(c_{2},c_{1}\right)$.

The function $f\left(t\right)$ is minimized when $\bm{u}\left(t\right)$
is orthogonal to $\bm{y}$. We have $\bm{y}^{T}\bm{u}_{1}\cos t_{\min}+\bm{y}^{T}\bm{u}_{2}\sin t_{\min}=0$.
Solving this equation proves part 4.

We prove the \emph{uniqueness} results next. By assumption, $\bm{y}$
is not orthogonal to both $\bm{u}_{1}$ and $\bm{u}_{2}$ simultaneously.
Hence, $\bm{y}_{p}\ne0$. Furthermore, since $\bm{u}_{1}$ and $\bm{u}_{2}$
are linearly independent, the vector $\bm{y}_{p}$ is uniquely defined.
This establishes the uniqueness of $t_{\max}$. Since the dimension
of the hyperplane $\mbox{span}\left(\left[\bm{u}_{1},\bm{u}_{2}\right]\right)$
is two, there exists a unique line in $\mbox{span}\left(\left[\bm{u}_{1},\bm{u}_{2}\right]\right)$
to be orthogonal to $\bm{y}_{p}\in\mbox{span}\left(\left[\bm{u}_{1},\bm{u}_{2}\right]\right)$.
We denote this line by a vector $\bm{y}_{\perp}\ne0$, such that $\bm{y}_{\perp}\in\mbox{span}\left(\left[\bm{u}_{1},\bm{u}_{2}\right]\right)$
and $\bm{y}_{\perp}^{T}\bm{y}_{p}=0$. First, $\bm{y}_{\perp}$ is
orthogonal to $\bm{y}$. This can be easily verified as $\bm{y}=\bm{y}_{p}+\bm{y}_{r}$,
where $\bm{y}_{r}$ is the projection residue vector and therefore
is orthogonal to $\bm{y}_{\perp}$ as well. Second, any linear combination
of $\bm{y}_{\perp}$ and $\bm{y}_{p}$ such that the coefficient of
$\bm{y}_{p}$ is nonzero produces a line that is not orthogonal to
$\bm{y}$. Therefore, $\bm{y}_{\perp}$ represents the unique line
in $\mbox{span}\left(\left[\bm{u}_{1},\bm{u}_{2}\right]\right)$ that
is orthogonal to $\bm{y}$. The corresponding value $t_{\min}$ is
therefore unique. 
\end{proof}
\vspace{0.01in}

We proceed next with the general case where $r\ge1$. Recall the expression
for the geodesic curve in (\ref{eq:geodesic-curve-rank1}). Denote
$\mathcal{P}_{\Omega}\left(\bm{h}\right)$ by $\bm{h}_{\Omega}$.
Similarly, we have $\bm{u}_{1,\Omega},\cdots,\bm{u}_{r,\Omega}$.
Let $\bm{u}_{1,\Omega}\left(t\right)=\bm{u}_{1,\Omega}\cos t+\bm{h}_{\Omega}\sin t$.
The atomic function can be written as \[
f\left(t\right)=\left\Vert \bm{x}_{\Omega}-\mathcal{P}\left(\bm{x}_{\Omega},\left[\bm{u}_{1,\Omega}\left(t\right),\bm{u}_{2,\Omega},\cdots,\bm{u}_{r,\Omega}\right]\right)\right\Vert _{F}^{2}.\]
 Again we drop the subscript $\Omega$ for convenience. The following
proposition is the key to understand the relationship between $\mathcal{P}\left(\bm{x},\bm{u}_{1}\left(t\right)\right)$
and $\mathcal{P}\left(\bm{x},\left[\bm{u}_{1}\left(t\right),\bm{u}_{2},\cdots,\bm{u}_{r}\right]\right)$. 
\begin{prop}
\label{pro:projection-high-rank}Let $\bm{y}\in\mathbb{R}^{m}$, $\bm{U}_{1}\in\mathbb{R}^{m\times n_{1}}$
and $\bm{U}_{2}\in\mathbb{R}^{m\times n_{2}}$ where $n_{1},n_{2}\in\left[m\right]$.
Let \[
\bm{y}_{r}=\bm{y}-\mathcal{P}\left(\bm{y},\left[\bm{U}_{1},\bm{U}_{2}\right]\right).\]
 Denote the $j^{th}$ column of $\bm{U}_{2}$ by $\left(\bm{U}_{2}\right)_{:j}$.
Then $\bm{y}_{r}$ can be written as \[
\bm{y}_{r}=\bm{y}_{r,1}-\mathcal{P}\left(\bm{y}_{r,1},\bm{U}_{2,r}\right),\]
 where $\bm{y}_{r,1}=\mathcal{P}\left(\bm{y},\bm{U}_{1}\right)$,
and $\bm{U}_{2,r}=\left[\left(\bm{U}_{2}\right)_{:1}-\mathcal{P}\left(\left(\bm{U}_{2}\right)_{:1},\bm{U}_{1}\right),\cdots,\left(\bm{U}_{2}\right)_{:r}-\mathcal{P}\left(\left(\bm{U}_{2}\right)_{:r},\bm{U}_{1}\right)\right]$. \end{prop}
\begin{proof}
The proof is centered around the notion of projection. For arbitrary
$\bm{y}\in\mathbb{R}^{m}$ and $\bm{U}\in\mathbb{R}^{m\times n}$,
an operator $\mathcal{P}$ is a projection operator if and only if
$\mathcal{P}\left(\bm{y},\bm{U}\right)\in\mbox{span}\left(\bm{U}\right)$
and $\bm{y}_{r}\perp\bm{U}$, where $\bm{y}_{r}=\bm{y}-\mathcal{P}\left(\bm{y},\bm{U}\right)$.
We say $\bm{y}_{r}\perp\bm{U}$ if $\bm{y}_{r}^{T}\bm{U}_{:j}=0$
for all $j\in\left[n\right]$.

Let $\bm{y}^{\prime}=\bm{y}_{r,1}-\mathcal{P}\left(\bm{y}_{r,1},\bm{U}_{2,r}\right)$.
To prove this proposition, it suffices to show that $\bm{y}_{r}^{\prime}\perp\left[\bm{U}_{1},\bm{U}_{2}\right]$
and $\bm{y}-\bm{y}_{r}^{\prime}\in\mbox{span}\left(\left[\bm{U}_{1},\bm{U}_{2}\right]\right)$.

We first show that $\bm{y}_{r}^{\prime}\perp\left[\bm{U}_{1},\bm{U}_{2}\right]$.
That $\bm{y}_{r}^{\prime}\perp\bm{U}_{1}$ is verified as follows.
Since $\mathcal{P}\left(\bm{y}_{r,1},\bm{U}_{2,r}\right)\in\mbox{span}\left(\bm{U}_{2,r}\right)$
and each column of $\bm{U}_{2,r}$ is orthogonal to $\bm{U}_{1}$,
we have $\mathcal{P}\left(\bm{y}_{r,1},\bm{U}_{2,r}\right)\perp\bm{U}_{1}$.
The definition of $\bm{y}_{r,1}$ implies that $\bm{y}_{r,1}\perp\bm{U}_{1}$.
Hence, we have $\bm{y}_{r}^{\prime}\perp\bm{U}_{1}$ as the vector
$\bm{y}_{r}^{\prime}$ is a linear combination of $\bm{y}_{r,1}$
and $\mathcal{P}\left(\bm{y}_{r,1},\bm{U}_{2,r}\right)$. We claim
that $\bm{y}_{r}^{\prime}\perp\bm{U}_{2}$ as well. According to the
definition of $\bm{y}_{r}^{\prime}$, it is clear that $\bm{y}_{r}^{\prime}\perp\bm{U}_{2,r}$.
Note that $\left(\bm{U}_{2}\right)_{:j}=\left(\bm{U}_{2,r}\right)_{:j}+\mathcal{P}\left(\left(\bm{U}_{2}\right)_{:j},\bm{U}_{1}\right)$.
The vector $\mathcal{P}\left(\left(\bm{U}_{2}\right)_{:j},\bm{U}_{1}\right)$
is in the $\mbox{span}\left(\bm{U}_{1}\right)$ and therefore orthogonal
to $\bm{y}_{r}^{\prime}$. As a result, $\bm{y}_{r}^{\prime}\perp\bm{U}_{2}$.
We then have $\bm{y}_{r}^{\prime}\perp\left[\bm{U}_{1},\bm{U}_{2}\right]$.

Next, we show that $\bm{y}-\bm{y}_{r}^{\prime}\in\mbox{span}\left(\left[\bm{U}_{1},\bm{U}_{2}\right]\right)$.
Note that \begin{align*}
\bm{y}-\bm{y}_{r}^{\prime} & =\bm{y}-\bm{y}_{r,1}+\mathcal{P}\left(\bm{y}_{r,1},\bm{U}_{2,r}\right)\\
 & =\mathcal{P}\left(\bm{y},\bm{U}_{1}\right)+\mathcal{P}\left(\bm{y}_{r,1},\bm{U}_{2,r}\right).\end{align*}
 Clearly, $\mathcal{P}\left(\bm{y},\bm{U}_{1}\right)\in\mbox{span}\left(\bm{U}_{1}\right)\subset\mbox{span}\left(\left[\bm{U}_{1},\bm{U}_{2}\right]\right)$.
Furthermore, according to the definition of $\bm{U}_{2,r}$, $\mbox{span}\left(\bm{U}_{2,r}\right)\subset\mbox{span}\left(\left[\bm{U}_{1},\bm{U}_{2}\right]\right)$
and therefore $\mathcal{P}\left(\bm{y}_{r,1},\bm{U}_{2,r}\right)\in\mbox{span}\left(\bm{U}_{2,r}\right)\subset\mbox{span}\left(\left[\bm{U}_{1},\bm{U}_{2}\right]\right)$.
This completes the proof. 
\end{proof}
\vspace{0.01in}

Based on the claim of this proposition, one can to apply the analysis
for the rank-one case (Proposition \ref{pro:atomic-fn-rank1}) to
higher-rank cases. Let $\bm{U}_{\sim1}=\left[\bm{u}_{2},\cdots,\bm{u}_{r}\right]$,
and let $\bm{x}_{r}=\bm{x}-\mathcal{P}\left(\bm{x},\bm{U}_{\sim1}\right)$.
Similarly, define $\bm{u}_{1,r}$ and $\bm{h}_{r}$. It is clear that
\begin{align*}
\bm{u}_{1,r}\left(t\right) & =\bm{u}_{1}\left(t\right)-\mathcal{P}\left(\bm{u}_{1}\left(t\right),\bm{U}_{\sim1}\right)\\
 & =\bm{u}_{1}\cos t+\bm{h}\sin t\\
 & \quad-\mathcal{P}\left(\bm{u}_{1},\bm{U}_{\sim1}\right)\cos t-\mathcal{P}\left(\bm{h},\bm{U}_{\sim1}\right)\sin t\\
 & =\bm{u}_{1,r}\cos t+\bm{h}_{r}\sin t.\end{align*}
 One has \begin{align*}
 & \bm{x}-\mathcal{P}\left(\bm{x},\left[\bm{u}_{1}\left(t\right),\bm{u}_{2},\cdots,\bm{u}_{r}\right]\right)\\
 & =\bm{x}_{r}-\mathcal{P}\left(\bm{x}_{r},\bm{u}_{1,r}\left(t\right)\right).\end{align*}
 This establishes the connection between the rank-one case and the
general case, proves Proposition \ref{pro:atomic-f-periodic}, and
justifies the procedure in Section \ref{sub:compute-tmin-tmax} for
computing minimizers and maximizers. 

\bibliographystyle{ieeetr}
\bibliography{MatrixCompletion}

\end{document}